\pdfoutput=1
\documentclass[thm-restate]{lipics-v2021}
\usepackage{style}
\usepackage{shortcuts}
\usepackage[capitalize]{cleveref}

\hideLIPIcs
\nolinenumbers

\crefname{lemma}{Lemma}{Lemmas}
 
\title{Improved Sublinear-Time Edit Distance for Preprocessed Strings}
\author{Karl Bringmann}{Saarland University and Max Planck Institute for Informatics, Saarland Informatics Campus,\\Saarbrücken, Germany}{}{}{}
\author{Alejandro Cassis}{Saarland University and Max Planck Institute for Informatics, Saarland Informatics Campus,\\Saarbrücken, Germany}{}{}{}
\author{Nick Fischer}{Saarland University and Max Planck Institute for Informatics, Saarland Informatics Campus,\\Saarbrücken, Germany}{}{}{}
\author{Vasileios Nakos}{RelationalAI, Berkeley, USA}{}{}{}
\authorrunning{K. Bringmann, A. Cassis, N. Fischer, V. Nakos}
\Copyright{Karl Bringmann, Alejandro Cassis, Nick Fischer, Vasileios Nakos}

\funding{This work is part of the project TIPEA that has received funding from the European Research Council (ERC) under the European Unions Horizon 2020 research and innovation programme (grant agreement No.~850979).}

\begin{document}
\maketitle

\begin{abstract}
We study the problem of approximating the edit distance of two strings in sublinear time, in a setting where one or both string(s) are preprocessed, as initiated by Goldenberg, Rubinstein, Saha~(STOC~'20). Specifically, in the \emph{$(k, K)$-gap edit distance} problem, the goal is to distinguish whether the edit distance of two strings is at most $k$ or at least~$K$. We obtain the following results:
\begin{itemize}
\item After preprocessing \emph{one} string in time $n^{1+\order(1)}$, we can solve $(k, k \cdot n^{\order(1)})$-gap edit distance in time~$(n/k + k) \cdot n^{\order(1)}$.
\item After preprocessing \emph{both} strings separately in time $n^{1+\order(1)}$, we can solve $(k, k \cdot n^{\order(1)})$-gap edit distance in time~$kn^{\order(1)}$.
\end{itemize}
Both results improve upon some previously best known result, with respect to either the gap or the query time or the preprocessing time.

Our algorithms build on the framework by Andoni, Krauthgamer and Onak (FOCS '10) and the recent sublinear-time algorithm by Bringmann, Cassis, Fischer and Nakos (STOC '22). We replace many complicated parts in their algorithm by faster and simpler solutions which exploit the preprocessing.
\end{abstract}

\section{Introduction}
The \emph{edit distance} (also known as \emph{Levenshtein distance}) is a fundamental measure of similarity between strings. It has numerous applications in several fields such as information retrieval, computational biology and text processing. Given strings $X$ and $Y$, their edit distance denoted by $\ED(X, Y)$ is defined as the minimum number of character insertions, deletions and substitutions needed to transform $X$ into $Y$.

A textbook dynamic programming algorithm computes the edit distance of two strings of length $n$ in time $O(n^2)$. Popular conjectures such as the \emph{Strong Exponential Time Hypothesis} imply that that this algorithm is essentially optimal, as there is no strongly subquadratic-time algorithm~\cite{BackursI15,AbboudBW15,BringmannK15,AbboudHWW16}. As for some applications involving enormous strings (such as DNA sequences) quadratic-time algorithms are impractical, a long line of research developed progressively better and faster \emph{approximation} algorithms~\cite{BarYossefJKK04,BatuES06,OstrovskyR07,AndoniO12,AndoniKO10,ChakrabortyDGKS20,KouckyS20,BrakensiekR20}. The current best approximation guarantee in near-linear time is an algorithm by Andoni and Nosatzki~\cite{AndoniN20} computing an $f(1/\varepsilon)$-approximation in time~$O(n^{1+\varepsilon})$.

Another more recent line of research studies edit distance in the \emph{sublinear-time} setting. Here the goal is to approximate the edit distance without reading the entire input strings. More formally, in the~$(k, K)$-gap edit distance problem the goal is to distinguish whether the edit distance between $X$ and $Y$ is at most $k$ or greater than $K$. The performance of gap algorithms is typically measured in terms of the string length $n$ and the gap parameters~$k$ and~$K$. This problem has been studied in several works~\cite{BatuEKMR03,AndoniO12,GoldenbergKS19,BrakensiekCR20,KociumakaS20} most of which focus on the~$(k, k^2)$-gap problem. Currently, there are two incomparable best known results: A recent result by Goldenberg, Kociumaka, Krauthgamer and Saha~\cite{GoldenbergKKS21} established a non-adaptive algorithm for the~$(k, k^2)$-gap problem in time $\widetilde{O}(n/k^{3/2})$.\footnote{We write $\widetilde\Order(\cdot)$ to hide polylogarithmic factors $(\log n)^{\Order(1)}$.} Another recent result by Bringmann, Cassis, Fischer and Nakos~\cite{BringmannCFN22} reduces the gap to $\Order^*(1)$ and solves the~$(k, \Order^*(k))$\=/gap problem in time $O^*(n/k + k^4)$.\footnote{We write $\Order^*(\cdot)$ to hide subpolynomial factors $n^{\order(1)}$ in $n$.} See \cref{tab:comparison} for a more detailed comparison.

\begin{table}[t]
    \caption{A comparison of sublinear-time algorithms for the $(k, k \cdot g)$-gap edit distance problem for different gap parameters $g$. All algorithms in this table are randomized and succeed with high probability. Note that some of these results are subsumed by others.} \label{tab:comparison}
    \begin{tabular*}{\linewidth}[t]{@{\extracolsep{\fill}}p{.4\linewidth}p{.07\linewidth}p{.2\linewidth}p{.14\linewidth}}
        \toprule
        \thead[l]{Source} & \thead[l]{Gap $g$} & \thead[l]{Preprocessing time} & \thead[l]{Query time} \\
        \midrule
        Goldenberg, Krauthgamer, Saha~\cite{GoldenbergKS19} & $\Order(k)$ & no preprocessing & $\widetilde\Order(n/k + k^3)$ \\
        Kociumaka, Saha~\cite{KociumakaS20} & $\Order(k)$ & no preprocessing & $\widetilde\Order(n/k + k^2)$ \\
        Brakensiek, Charikar, Rubinstein~\cite{BrakensiekCR20} & $\Order(k)$ & no preprocessing & $\widetilde\Order(n/\sqrt k)$ \\
        Bringmann, Cassis, Fischer, Nakos~\cite{BringmannCFN22} & $\Order(k)$ & no preprocessing & $\Order^*(n/k^2 + k^8)$ \\
        Goldenberg, Kociumaka, Krauthgamer, Saha~\cite{GoldenbergKKS21} & $\Order(k)$ & no preprocessing & $\widetilde\Order(n/k^{3/2})$ \\
        Bringmann, Cassis, Fischer, Nakos~\cite{BringmannCFN22} & $\Order^*(1)$ & no preprocessing & $\Order^*(n/k + k^4)$ \\[2ex]
        Goldenberg, Rubinstein, Saha~\cite{GoldenbergRS20} & $\Order(k)$ & one-sided, $\widetilde\Order(n)$ & $\widetilde\Order(n/k + k^2)$ \\
        Brakensiek, Charikar, Rubinstein~\cite{BrakensiekCR20} & $g$ & one-sided, $\widetilde\Order(nk / g)$ & $\widetilde\Order(n/g + k^2/g)$ \\
        \emph{This work, \cref{thm:one-sided}} & $\Order^*(1)$ & one-sided, $\Order^*(n)$ & $\Order^*(n/k + k)$ \\[2ex]
        Chakraborty, Goldenberg, Kouck\'{y}~\cite{ChakrabortyGK16} & $\Order(k)$ & two-sided, $\widetilde\Order(n)$ & $\Order(\log n)$ \\
        Brakensiek, Charikar, Rubinstein~\cite{BrakensiekCR20} & $g$ & two-sided, $\widetilde\Order(nk / g)$ & $\widetilde\Order(k^2 / g)$ \\
        Ostrovsky, Rabani~\cite{OstrovskyR07} & $\Order^*(1)$ & two-sided, $\widetilde\Order(n^2)$ & $\Order(\log n)$ \\
        \emph{This work, \cref{thm:two-sided}} & $\Order^*(1)$ & two-sided, $\Order^*(n)$ & $\Order^*(k)$ \\
        Goldenberg, Rubinstein, Saha~\cite{GoldenbergRS20} & $\Order(1)$ & two-sided, $\widetilde\Order(n^2)$ & $\Order^*(n^{3/2})$ \\
        Goldenberg, Rubinstein, Saha~\cite{GoldenbergRS20} & $1$ & two-sided, $\widetilde\Order(n)$ & $\widetilde\Order(k^2)$ \\
        \bottomrule
    \end{tabular*}
\end{table}

Our starting point is the work by Goldenberg, Rubinstein and Saha~\cite{GoldenbergRS20} which studies sublinear algorithms for edit distance in the \emph{preprocessing model}. Here, we are allowed to preprocess one or both input strings $X$ and $Y$ \emph{separately}, and then use the precomputed information to solve the $(k, K)$-gap edit distance problem. This model is motivated by applications where many long strings are compared against each other. For example, the \emph{string similarity join} problem is to find all pairs of strings in a database (containing e.g.\ DNA sequences) which are close in edit distance; see~\cite{WandeltDGMMPSTWWL14} for a survey on practically relevant algorithms. Note that in these applications, if we have an algorithm with almost-linear preprocessing time (which is the case for all the algorithms we present in this paper), then the overhead incurred by preprocessing is comparable to the time necessary to read and store the strings in the first place. In~\cite{GoldenbergRS20}, the authors pose and investigate the following open question:
\begin{center}
    \emph{``What is the complexity of approximate edit distance\\with preprocessing when $k \ll n$?''~\cite{GoldenbergRS20}}
\end{center}
This question has spawned significant interest in the community~\cite{ChakrabortyGK16,GoldenbergRS20,BrakensiekCR20}, and with this paper we also make progress towards this question. We give an overview of results in \cref{tab:comparison}. Note that most results are hard to compare to each other (one-sided versus two-sided preprocessing, exact versus $\Order(1)$-approximate versus $\Order(k)$-approximate).

In the two-sided model, all known algorithms (with almost-linear preprocessing time\footnote{Here we insist on almost-linear preprocessing time since the celebrated embedding of edit distance into the $\ell_1$-metric with distortion $n^{o(1)}$ due to Ostrovsky and Rabani~\cite{OstrovskyR07} achieves query time $O(\log n)$ but requires preprocessing time $\Omega(n^2)$.} and, say, subpolynomial gap $g = n^{\order(1)}$) share the common barrier that the query time is~$\Omega(k^2)$. Due to this barrier, Goldenberg et al.~\cite{GoldenbergRS20} specifically ask whether there exists an approximation algorithm with sub-$k^2$ query time. One of our contributions is that we answer this question in the affirmative.

\subparagraph*{Our Results}
We develop sublinear-time algorithms for the~$(k, \Order^*(k))$-gap edit distance problem, in the one-sided and two-sided preprocessing model, respectively.

\begin{theorem}[One-Sided Preprocessing] \label{thm:one-sided}
Let $X, Y$ be length-$n$ strings. After preprocessing~$Y$ in time $\Order^*(n)$, we can solve the $(k, k \cdot n^{o(1)})$-gap edit distance problem for $X$ and $Y$ in time~$\Order^*(n/k + k)$ with high probability.
\end{theorem}

In comparison to the $(k, k^2)$-gap algorithms from~\cite{GoldenbergRS20,BrakensiekCR20} with best query time\footnote{For the $(k, k^2)$-gap problem, the running time bounds $\widetilde\Order(n/k + k)$ and $\widetilde\Order(n/k)$ can be considered equal, as for~$k \geq \sqrt n$ the algorithm may return a trivial answer.}~$\widetilde\Order(n/k + k)$, we contribute the following improvement: Ignoring lower-order factors, we reduce the gap from~$k$ to~$\Order^*(1)$ while achieving the same query time $\Order^*(n/k + k)$ and the same preprocessing time $\Order^*(n)$. In comparison to the $(k, \Order^*(k))$-gap algorithm in time~\makebox{$\Order^*(n/k + k^4)$} from~\cite{BringmannCFN22}, we achieve the same gap but an improved query time for large~$k$, at the cost of preprocessing one of the strings.

\medskip
In the two-sided model, we obtain an analogous result, where the query time no longer depends on $n/k$.

\begin{theorem}[Two-Sided Preprocessing] \label{thm:two-sided}
Let $X, Y$ be length-$n$ strings. After preprocessing both $X$ and $Y$ (separately) in time $\Order^*(n)$, we can solve the $(k, k \cdot n^{o(1)})$-gap edit distance problem for $X$ and $Y$ in time~$\Order^*(k)$ with high probability. 
\end{theorem}

We remark that all hidden factors in both theorems are $2^{\widetilde \Order(\sqrt{\log n})}$. For a detailed comparison of this algorithm to the previously known results, see \cref{tab:comparison}. We point out that \cref{thm:two-sided} settles the open question from~\cite{GoldenbergRS20} whether there exists an edit distance approximation algorithm with small gap and sub-$k^2$ query time.

\subparagraph*{Our Techniques}
To achieve our results we build on the recent sublinear-time algorithm by Bringmann, Cassis, Fischer and Nakos~\cite{BringmannCFN22}, which itself builds on an almost-linear-time algorithm by Andoni, Krauthgamer and Onak~\cite{AndoniKO10}. The basic idea of the original algorithm is to split the strings into several smaller parts and recur on these parts with non-uniform precisions. The idea of~\cite{BringmannCFN22} is to prune branches in the recursion tree, by detecting and analyzing periodic substructures. Towards that, they designed appropriate property testers to efficiently detect these structures. In our setting, we observe that having preprocessed the string(s), we can prune the computation tree much more easily. Thus, our algorithm proceeds in the same recursive fashion as~\cite{BringmannCFN22} and uses a similar set of techniques, but due to the preprocessing it turns out to be simpler and faster.

\subparagraph*{Further Related Work}
In the previous comparison about sublinear-time algorithms, we left out streaming and sketching algorithms~\cite{BelazzouguiZ16,ChakrabortyGK16,BarYossefJKK04} and document exchange protocols~\cite{Jowhari12,BelazzouguiZ16,Haeupler19}.

\subparagraph*{Future Directions}
There are several interesting directions for future work. We specifically mention two open problems.

\begin{enumerate}
    \item \emph{Constant gap?} As~\cref{tab:comparison} shows, so far no constant-gap sublinear-time algorithm is known. Maybe the one-sided preprocessing setting is more approachable for this challenge. We believe that our approach is hopeless to achieve a constant gap, since we borrow from the recursive decomposition introduced in~\cite{AndoniKO10} which inherently incurs a polylogarithmic overhead in the approximation factor.
    \item \emph{Improving the query time?} The well-known $\Omega(n/K)$ lower bound against the $(k, K)$-gap \emph{Hamming distance} problem (and therefore against edit distance) continues to hold in the one-sided preprocessing setting. In particular, the most optimistic hope is an algorithm with query time $\Order^*(n/k)$ for the $(k, \Order^*(k))$-gap edit distance problem. Can this be achieved or is the extra $+k$ in the query time of~\cref{thm:one-sided} necessary? For two-sided preprocessing, to the best of our knowledge no lower bound is known.
\end{enumerate}
\section{Preliminaries} \label{sec:preliminaries}
We set $\range ij = \set{i, i+1, \dots, j-1}$ (in particular, $\range ii = \emptyset$) and $\rangezero j = \range 0j$. We say that an event happens \emph{with high probability} if it happens with probability at least~\makebox{$1 - 1/\poly(n)$}, where the degree of the polynomial can be an arbitrary constant. We write $\poly(n) = n^{\Order(1)}$ and $\widetilde\Order(n) = n (\log n)^{\Order(1)}$.

Let $X, Y$ be strings over an alphabet $\Sigma$ with polynomial size. We denote by $|X|$ the length of $X$. We denote by $X \circ Y$ the \emph{concatenation} of $X$ and $Y$. We denote by~$X \access i$ the $i$-th character in $X$ starting with index zero. We denote by $X \range ij$ the substring of $X$ with indices in $\range ij$, that is, including~$i$ and excluding~$j$. For out-of-bounds indices we set~\makebox{$X \range ij = X \range{\max(i, 0)}{\min(j, |X|)}$}. If~$X$ and~$Y$ have the same length, we define their \emph{Hamming distance $\HD(X, Y)$} as the number of non-matching characters~\makebox{$\HD(X, Y) = |\{\, i : X \access i \neq Y \access i \,\}|$}. For two strings~$X, Y$ with possibly different lengths, we define their \emph{edit distance $\ED(X, Y)$} as the smallest number of character \emph{insertions}, \emph{deletions} and \emph{substitutions} necessary to transform $X$ into~$Y$. An \emph{optimal alignment} between~$X$ and~$Y$ is a monotonically non-decreasing function~\makebox{$A : \set{0, \dots, |X|} \to \set{0, \dots, |Y|}$} such that~$A(0) = 0$, $A(|X|) = |Y|$ and
\begin{equation*}
    \ED(X, Y) = \sum_{i=0}^{|X|-1} \ED(X \access i, Y \range{A(i)}{A(i+1)}).    
\end{equation*}
It is easy to see that there is an optimal alignment between any two strings $X, Y$: Trace an optimal path through the standard dynamic program for edit distance and assign $A(i)$ to the smallest $j$ for which the path crosses~$(i, j)$.

Let $T$ be a rooted tree, and let $v$ be a node in~$T$. We denote by $\Root(T)$ the root node in~$T$. We denote by $\parent(v)$ the parent node of $v$. We denote by $\depth(v)$ the length of the root-to-$v$ path, and by $\height(v)$ the length of the longest $v$-to-leaf path.
\section{Overview} \label{sec:overview}

\subsection{A Linear-Time Algorithm à la Andoni-Krauthgamer-Onak}
We start to outline an almost-linear-time algorithm to approximate the edit distance of two strings~$X, Y$ following the framework of the Andoni-Krauthgamer-Onak algorithm~\cite{AndoniKO10}, with some changes as in~\cite{BringmannCFN22} and some additional modifications (see the novel trick outlined at the end of this subsection).

\subparagraph*{First Ingredient: A Divide-and-Conquer Scheme}
The basic idea of the algorithm is to apply a divide-and-conquer scheme to reduce the approximation of the global edit distance to approximating the edit distance of several smaller strings. The straightforward idea of partitioning both strings $X, Y$ into parts~$X_1, \dots, X_m, Y_1, \dots, Y_m$ and computing the edit distances $\ED(X_i, Y_i)$ does not immediately work; instead we need to consider \emph{several shifts} of the string $Y_i$. We remark that this concept of recurring on smaller strings for several shifts is quite standard in previous work. The following lemma uses the same ideas as the ``$\mathcal E$-distance'' defined in~\cite{AndoniKO10}. We give a proof in \cref{sec:divide-and-conquer}.

\begin{restatable}[Divide and Conquer]{lemma}{lemdivideandconquer} \label{lem:divide-and-conquer}
Let $X, Y$ be length-$n$ strings, and let $0 = j_0 < \dots < j_B = n$. We write $X_i = X \range{j_{i-1}}{j_i}$ and $Y_{i, s} = Y \range{j_{i-1}+s}{j_i+s}$.
\begin{itemize}
\item For all shifts $s_1, \dots, s_B$ we have that $\ED(X, Y) \leq \sum_i \ED(X_i, Y_{i, s_i}) + 2|s_i|$.
\item There are shifts $s_1, \dots, s_B$ with $\sum_i \ED(X_i, Y_{i, s}) \leq 2\ED(X, Y)$ and $2|s_i| \leq \ED(X, Y)$ for all~$i$.
\end{itemize}
\end{restatable}
\medskip

To explain how to apply \cref{lem:divide-and-conquer}, we first specify on which substrings our algorithm is supposed to recur. To this end, let $T$ be a balanced $B$-ary tree with~$n$ leaves. $T$ will act as the~``recursion tree'' of the algorithm. For a string $X$ of length~$n$, we define a substring~$X_v$ for every node $v$ in $T$ as follows: If the subtree below $X_v$ spans from the $i$-th to the $j$-th leaf (ordered from left to right), then we set~$X_v = X \range{i}{j+1}$. In particular, $X_v$ is a single character for each leaf $v$, and~$X_{\Root(T)} = X$. We further define~$X_{v, s} = X \range{i + s}{j + s + 1}$. For concreteness, we set $B = 2^{\sqrt{\log n \log \log n}} = n^{o(1)}$ throughout the paper.

\subparagraph*{A Simple Algorithm}
Based on \cref{lem:divide-and-conquer}, we next present a simple (yet slow) algorithm. Our goal is to compute, for each node $v$ in the tree $T$, an approximation~\raisebox{0pt}[0pt][0pt]{$\widetilde\ED(X_v, Y_{v, s})$} of~$\ED(X_v, Y_{v, s})$ for all shifts $s$. The result at the root node is returned as the desired approximation of $\ED(X, Y)$. The algorithm works as follows: For each leaf we can cheaply compute~$\ED(X_v, Y_{v, s})$ exactly by comparing the single characters~$X_v$ and~$Y_{v, s}$. For each internal node with children~$v_1, \dots, v_B$ we compute
\begin{equation} \label{eq:divide-and-conquer}
    \widetilde\ED(X_v, Y_{v, s}) = \sum_{i=1}^B \min_{s_i \in \Int} \widetilde\ED(X_{v_i}, Y_{v_i, s_i}) + 2 |s - s_i|.
\end{equation}
A careful application of \cref{lem:divide-and-conquer} shows that if the recursive approximations $\ED(X_{v_i}, Y_{v_i, s_i})$ have multiplicative error at most $\alpha$, then by approximating $\ED(X_v, Y_{v, s})$ as in~\eqref{eq:divide-and-conquer} the multiplicative error becomes~$2\alpha + B$. Since we repeat this argument recursively up to depth $\depth(T) \leq \log_B(n)$, the multiplicative error accumulates to $B \cdot \exp(\Order(\log_B(n))) = n^{\order(1)}$.

This simple algorithm achieves the desired approximation quality, however, it is not fast enough: For every node $v$ we have to compute $\widetilde\ED(X_v, Y_{v, s})$ for too many shifts $s$ (naively speaking, for up to $n$ shifts). As a first step towards dealing with this issue, we first show that at every node $v$ we can in fact tolerate a certain \emph{additive} error (in addition to the multiplicative error discussed before) using a technique called \emph{Precision Sampling}. Then we exploit the freedom of additive errors to run this algorithm for a \emph{restricted} set of shifts $s$.

\subparagraph*{Second Ingredient: Precision Sampling}
It ultimately suffices to compute an approximation of $\ED(X, Y)$ with \emph{additive} error $k$ in order to solve the constant-gap edit distance problem. We leverage this freedom to also solve the recursive subproblems up to some additive error. Specifically, we will work with the following data structure:

\begin{definition}[Precision Tree]
Let $T$ be a balanced $B$-ary tree with $n$ leaves. For~\makebox{$t \in \Nat$}, we randomly associate a \emph{tolerance $t_v$} to every node $v$ in $T$ as follows:
\begin{itemize}
\item If $v$ is the root, then set $t_v = t$;
\item otherwise set $t_v = t_{\mathsf{parent}(v)} \cdot u_v / 3$, where we sample $u_v \sim \Exp(\Order(\log n))$ (the exponential distribution with parameter $\Order(\log n)$).
\end{itemize}
We refer to $T$ as a \emph{precision tree with initial tolerance $t$}.
\end{definition}

The tolerance $t_v$ at a node $v$ determines the additive error which we can tolerate at~$v$. That is, our goal is to approximate $\ED(X_v, Y_{v, s})$ with additive error $t_v$ (and the same multiplicative error as before). The initial tolerance is set to $t = k$. The critical step is how to combine the recursive approximations with additive error $t_{v_1}, \dots, t_{v_B}$ to an approximation with additive error $t_v$. The naive solution would incur error $\sum_i t_{v_i} \gg t_v$. Instead, we employ the \emph{Precision Sampling Lemma}~\cite{AndoniKO10,AndoniKO11,Andoni17,BringmannCFN22} (see \cref{lem:psl} in \cref{sec:proofs}) to recombine the recursive approximation and avoid this blow-up in the additive error.

\subparagraph*{An Improved Algorithm}
We can now improve the simple algorithm to a near-linear-time algorithm. In the original Andoni-Krauthgamer-Onak algorithm this was achieved by pruning most recursive subproblems (depending on their tolerances $t_v$). We will follow a different avenue: Our algorithm recurs on every node in the precision tree, and we obtain a linear-time algorithm by bounding the expected running time per node by $n^{\order(1)}$.

We achieve this by the following novel trick: We restrict the set of feasible shifts $s$ at each node $v$ with respect to the tolerance $t_v$. In fact, we require two constraints: First, we restrict~$s$ to values smaller than $\approx k$ in absolute value. This first restriction is correct since we only want to maintain edit distances bounded by $k$; this idea was also used in previous works. Second, we restrict the feasible shifts $s$ at any node $v$ to multiples of~$\floor{t_v / 2}$. Then, in order to approximate $\ED(X_v, Y_{v, s})$ for any shift $s$ we let $\tilde s$ denote the closest multiple of~$\floor{t_v / 2}$ to $s$, and approximate $\ED(X_v, Y_{v, s})$ by $\ED(X_v, Y_{v, \tilde s})$. Since $|s - \tilde s| \leq t_v / 2$, both edit distances differ by at most $t_v$. Recall that we can tolerate this error using the Precision Sampling technique. Let $S_v = \set{-k \cdot n^{\order(1)}, \dots, k \cdot n^{\order(1)}} \cap \floor{t_v/2} \Int$ denote the set of shifts respecting these restrictions (the precise lower-order term $n^{\order(1)}$ will be fixed later).

In terms of efficiency, we have improved as follows: At every node the running time is essentially dominated by the number of feasible shifts $s$. Using our discretization trick, there are only $|S_v| = k \cdot n^{\order(1)} / t_v$ such shifts. By the following \cref{lem:expected-precision}, we can bound this number in expectation by~\makebox{$k \cdot n^{\order(1)} / t_{\Root(T)} = n^{\order(1)}$}.

\begin{restatable}[Expected Precision]{lemma}{lemexpectedprecision} \label{lem:expected-precision}
Let $T$ be a $B$-ary precision tree with initial tolerance $t$ and $B = \exp(\widetilde \Theta(\sqrt{\log n}))$, and let $v$ be a node in $T$. Then, conditioned on a high-probability event $E$, it holds that
\begin{equation*}
    \Ex\left(\frac1{t_v} \;\middle|\; E\right) \leq \frac{(\log n)^{\Order(\depth(v))}}{t} \leq \frac{n^{\order(1)}}{t}.
\end{equation*}
\end{restatable}

We include a proof in \cref{sec:divide-and-conquer}. For technical reasons, the lemma is only true conditioned on some high-probability event $E$. For the remainder of this paper, we implicitly condition on this event~$E$.

\subsection{How to Go Sublinear?}
Following the idea of~\cite{BringmannCFN22}, our strategy is to turn this algorithm into a sublinear-time algorithm by not exploring the whole precision tree recursively, and instead only exploring a smaller fragment. To achieve this, the goal is to approximate $\ED(X_v, Y_{v, s})$ for many nodes $v$ directly, without the need to explore their children---we say that we \emph{prune} $v$. The algorithm by~\cite{BringmannCFN22} uses several structural insights on periodic versus non-periodic strings to implement pruning rules. We can avoid the complicated treatment and follow a much simpler avenue, exploiting that we can preprocess the strings.

In the following, we will assume that we have access to an oracle answering the following two queries. In the next section we will argue how to efficiently implement data structures to answer these queries.
\begin{itemize}
\smallskip
\item $\Matching(X, Y, v)$: Returns either $\Close[s^*]$ where $s^*$ satisfies $|s^*| \leq k \cdot n^{\order(1)}$ and $\HD(X_v, Y_{v, s^*}) \leq t_v / 2$, or $\Far$ in case that there is no shift $s^*$ with \makebox{$X_v = Y_{v, s^*}$}, see \cref{def:matching-queries}. (Note that if~$1 \leq \min_{s^*} \HD(X_v, Y_{v,s^*}) \leq t_v/2$, the query can return either $\Close[s^*]$ or $\Far$.)
\smallskip
\item $\ShiftedDistance(Y, v, s, s')$: For shifts $|s|, |s'| \leq k \cdot n^{\order(1)}$, returns an approximation of $\ED(Y_{v, s}, Y_{v, s'})$ with additive error $t_v / 2$ (and $n^{o(1)}$-multiplicative error, see \cref{def:shifted-distance-queries}).
\smallskip
\end{itemize}
Suppose for the moment that both queries can be answered in constant time. Then we can reduce the running time of the previous linear-time algorithm to time $k \cdot n^{\order(1)}$ as follows: We try to prune each node $v$ by querying $\Matching(X, Y, v)$. If the \Matching{} query reports~\Far, we simply continue recursively as before (i.e., no pruning takes place). However, if the matching query reports $\Close[s^*]$, then we can prune $v$ as follows: Query $\ShiftedDistance(Y, v, s^*, s)$ for all shifts $s \in S_v$ and return the outcome as an approximation of~$\ED(X_v, Y_{v, s})$ for all shifts $s$. For the correctness we apply the triangle inequality and argue that the additive error is bounded by $t_v / 2 + t_v / 2 = t_v$.

It remains to argue that the number of recursive computations is bounded by $k \cdot n^{\order(1)}$. The intuitive argument is as follows: Assume that $\ED(X, Y) \leq k$ (i.e., we are in the~``close'' case) and consider an optimal alignment between $X$ and $Y$, which contains at most $k$ mismatches. Recall that each level of the precision tree induces a partition of $X$ into consecutive substrings~$X_v$. Thus, there are at most $k$ substrings $X_v$ which contain a mismatch in the optimal alignment. For all \emph{other} substrings, there are no mismatches and hence $X_v = Y_{v, s^*}$ for some shift $s^*$. It follows that on every level there are at most $k$ nodes for which the~\Matching{} test fails, and in total there are only $k \cdot \height(T) = \Order(k \log n)$ such nodes.

\subsection{How to Answer Matching and Shifted-Distance Queries?}
It remains to find data structures which answer these queries efficiently. We assume that the precision tree $T$ has been generated in advance and is shared across all precomputations. In particular, this requires ``public randomness'' for the otherwise independent precomputations.

\subparagraph*{Matching Queries}
The idea is to precompute and store fingerprints (i.e.\ hashes) of the substrings~$Y_{v,s}$ for every node $v$ in the partition tree and all shifts $s \in \set{-k \cdot n^{\order(1)}, \dots, k \cdot n^{\order(1)}}$. Then, to answer a query we simply compute the fingerprint of $X_v$ and lookup whether it equals one of the precomputed fingerprints. Alas, a naive implementation of this idea is too slow since upon query we might need to read the whole string $X$. To obtain the desired sublinear query time, we instead subsample the strings $X_v$ and $Y_{v,s}$ with rate $\approx 1/t_v$. In this way we incur additive Hamming error at most $t_v / 2$, as desired. Formally, we show the following lemma in \cref{sec:proofs:sec:queries}:

\begin{restatable}[Matching Queries]{lemma}{lemmatchingqueries}\label{lem:matching-queries}
We can preprocess $Y$ in expected time $n^{1+\order(1)}$ to answer $\Matching(X, Y, v)$ queries in time $\widetilde\Order(|X_v| / t_v)$ with high probability. Moreover, we can separately preprocess both $X$ and $Y$ in expected time $n^{1+\order(1)}$ to answer $\Matching(X, Y, v)$ queries in time $\Order(1)$ with high probability.
\end{restatable}

The key difference between the one-sided and the two-sided preprocessing is that in the former we need to compute the fingerprint for $X_v$, which takes time $\Order(|X_v|/t_v)$ (we shave the factor $t_v$ due to the subsampling), while in the latter we can afford to precompute these fingerprints and answer the queries faster.

\subparagraph*{Shifted-Distance Queries}
We give two ways to answer \ShiftedDistance{} queries. The first one relies in a black-box manner on any almost-linear-time algorithm to compute an edit distance approximation with multiplicative error $n^{\order(1)}$~\cite{AndoniKO10,AndoniO12,AndoniN20}. Using the same trick as before to restrict the set of feasible shifts, it suffices to approximate $\ED(Y_{v, s}, Y_{v, s'})$ for all shifts~$s \in S_v$ and~$s' \in \set{-k \cdot n^{\order(1)}, \dots, k \cdot n^{\order(1)}}$. (We could also discretize the range of $s'$, but this does not improve the performance here.) In this way, we incur an additive error of at most $\Order(t_v)$. The computation per node $v$ takes time $|Y_v| \cdot k \cdot k \cdot n^{\order(1)} / t_v$, which becomes~$k n^{1+\order(1)}$ in expectation by~\cref{lem:expected-precision} and by summing over all nodes $v$. 

We then show that we can improve the preprocessing time to $n^{1+o(1)}$ by applying the non-oblivious embedding of edit distance into $\ell_1$ by Andoni and Onak~\cite{AndoniO12}. Formally we obtain the following lemma, which we prove in~\cref{sec:proofs:sec:fast-preprocessing}.

\begin{restatable}[Shifted-Distance Queries]{lemma}{lemshifteddistancequeries}\label{lem:shifted-distance-queries}
We can preprocess $Y$ in time $n^{1+\order(1)}$ to answer $\ShiftedDistance(Y, v, s, s')$ queries in time $n^{o(1)}$ with high probability.
\end{restatable}
\section{Our Algorithm in Detail} \label{sec:proofs}
In this section we give a detailed proof of our main theorems by analyzing \cref{alg:pruning-rule,alg:main} (see \cref{sec:overview:sec:pruning-rule,sec:overview:sec:main}). \cref{alg:main} is the previously discussed reformulation of the Andoni-Krauthgamer-Onak algorithm, with the improvement that the recursive computation can be avoided whenever \cref{alg:pruning-rule} succeeds. \cref{alg:pruning-rule} implements the pruning rule which, as we will argue, triggers often enough to improve the running time.

Throughout we fix the initial tolerance of the precision tree~$T$ to $t = t_{\Root(T)} = k$. Moreover, we set $S = \set{-k\cdot 3^{\log_B(n)}, \dots, k \cdot 3^{\log_B(n)}}$ and $S_v = S \cap \floor{t_v / 2} \Int$. Observe that~$|S| = k \cdot 3^{\log_B(n)} = k \cdot n^{\order(1)}$ for our choice of $B = \exp(\widetilde\Order(\sqrt{\log n}))$.

\subparagraph*{Outline}
Compared to the technical overview, we present the details in the reverse order: Starting with the implementation of the data structures for \Matching{} and \ShiftedDistance{} queries (in \cref{sec:proofs:sec:queries}), we first analyze \cref{alg:pruning-rule} (in \cref{sec:overview:sec:pruning-rule}). Then we analyze \cref{alg:main} (in \cref{sec:overview:sec:main}) and put the pieces together for our main theorems (in \cref{sec:overview:sec:assemble}).

\subsection{Matching Queries} \label{sec:proofs:sec:queries}
In this section we show to implement data structures that answer \Matching{} queries. We start with a formal definition.

\begin{definition}[Matching Queries] \label{def:matching-queries}
Let $X, Y$ be strings of length $n$, and let $v$ be a node in the precision tree. A $\Matching(X, Y, v)$ query is correctly answered by one of two outputs:
\begin{itemize}
\item $\Close[s^*]$ for some shift $s^* \in S$ satisfying $\HD(X_v, Y_{v, s^*}) \leq t_v/2$, or
\item $\Far$ if there exists no shift $s^* \in S$ with $X_v = Y_{v, s^*}$.
\end{itemize}
\end{definition}

\lemmatchingqueries*
\begin{proof}
We first explain the general idea behind the data structure and then point out the specifics for the one-sided and two-sided preprocessing. For each node $v$ in the precision tree, we subsample a set $H_v \subseteq \rangezero{|Y_v|}$ with rate $\Theta(\log n / t_v)$ and sample a hash function~\makebox{$h : \Sigma^{|H_v|} \to \rangezero{\poly (n)}$} from any universal family of hash functions. (Take for instance the function $h(\sigma_1, \dots, \sigma_{|H_v|}) = \sum_i a_i \sigma_i \bmod p$ for some prime $p = \poly (n)$ and random $a_i$'s).

We associate to a string $A \in \Sigma^{|Y_v|}$ the \emph{fingerprint} $h(A \access{H_v})$, where we write $A \access{H_v}$ for the subsequence of $A$ with indices in $H_v$. We claim that these fingerprints can distinguish any two strings $A, B$ with $\HD(A, B) > t_v / 2$ with high probability. Indeed, the probability that $H_v$ contains an index $i$ with $A \access i \neq B \access i$ is at least
\begin{equation*}
    1 - \left(1 - \frac{\Omega(\log n)}{t_v}\right)^{\HD(A, B)} \geq 1 - \exp(-\Omega(\log n)) = 1 - \frac1{\poly (n)}.
\end{equation*}
Hence, with high probability we have that $A \access{H_v} \neq B \access{H_v}$. Moreover, since $h$ is sampled from a universal family of hash functions, we also have that $h(A \access{H_v}) \neq h(B \access{H_v})$ with high probability. On the other hand, if $A = B$ then clearly $h(A \access{H_v}) = h(B \access{H_v})$. We now turn to the implementation details for one-sided and two-sided preprocessing.

\proofsubparagraph{One-Sided Preprocessing.}
In the preprocessing phase, we prepare for every node $v$ the fingerprints of $Y_{v, s}$ for all shifts $s \in S$, and store these fingerprints in a lookup table. In the query phase, given some string $X$ we compute the fingerprint of $X_v$ and check whether it appears in the lookup table. If so, we return $\Close[s^*]$ for the shift $s^*$ corresponding to the precomputed fingerprint. Otherwise, we return $\Far$. The correctness follows from the previous paragraph.

The expected running time of the preprocessing phase is bounded by $\sum_v \widetilde\Order(|Y_v| / t_v \cdot |S|)$ (for every node we have to prepare $|S|$ fingerprints, each taking time $\widetilde\Order(|Y_v| / t_v$)). This becomes $\sum_v \Order(|Y_v| \cdot n^{\order(1)} / t_{\Root(T)} \cdot k) = \sum_v \Order(|Y_v| \cdot n^{\order(1)}) = n^{1+\order(1)}$ in expectation over the tolerances $t_v$; see \cref{lem:expected-precision}. The query time is dominated by computing the fingerprint of~$X_v$ which takes expected time $\widetilde\Order(|X_v| / t_v)$ (even with high probability by an application of the Chernoff bound).

\proofsubparagraph{Two-Sided Preprocessing.}
For two-sided preprocessing, we can also prepare the fingerprints of $X_v$ for all nodes $v$ in the preprocessing phase. The expected preprocessing time is still~$n^{1+\order(1)}$, but for queries it only takes constant time to perform the lookup. 
\end{proof}

\subsection{Simple Shifted-Distance Queries}
We next demonstrate how to deal with \ShiftedDistance{} queries, formally defined as follows:

\begin{definition}[Shifted-Distance Queries] \label{def:shifted-distance-queries}
Let $Y$ be a string of length $n$, let $v$ be a node in the precision tree and let $s, s' \in S$. A $\ShiftedDistance(Y, v, s, s')$ query computes a number satisfying:
\begin{equation*}
    \ED(Y_{v, s}, Y_{v, s'}) - t_v / 2 \leq \ShiftedDistance(Y, v, s, s') \leq n^{\order(1)} \cdot (\ED(Y_{v, s}, Y_{v, s'}) + t_v / 2).
\end{equation*}
\end{definition}

We will first show how to implement a simpler version of~\cref{lem:shifted-distance-queries} at the cost of worsening the preprocessing time to $k n^{1+\order(1)}$ instead of $n^{1+\order(1)}$. The benefit of this weaker version is that it is a black-box reduction to any almost-linear time $n^{o(1)}$-approximation algorithm for edit distance, while the improved version crucially relies on the properties of the particular algorithm by Andoni and Onak~\cite{AndoniO12}. In the next section we show how to obtain the speed-up.

\begin{lemma}[Slower Shifted-Distance Queries]
We can preprocess $Y$ in time $kn^{1+\order(1)}$ to answer $\ShiftedDistance(Y, v, s, s')$ queries in time $n^{o(1)}$ with high probability.
\end{lemma}
\begin{proof}
We will use the result that an $n^{\order(1)}$-approximation for the edit distance of two length-$n$ strings can be computed in time $n^{1+o(1)}$~\cite{AndoniKO10,AndoniO12,AndoniN20}. 
In the preprocessing phase, we compute $n^{\order(1)}$-factor approximations \raisebox{0pt}[0pt][0pt]{$\widetilde{\ED}(Y_{v,\tilde s}, Y_{v,s'})$} for all nodes $v$, all~\makebox{$\tilde s \in S_v$} and all~\makebox{$s' \in S$}. Then, to answer a query $\ShiftedDistance(Y, v, s, s')$ we let
\begin{equation*}
    \tilde s := \argmin_{\tilde{s} \in S_v} |s-\tilde{s}|    
\end{equation*}
and output $\widetilde{\ED}(Y_{v, \tilde s}, Y_{v,s'})$.

First we argue that this gives a good approximation. Indeed, we have that $|s - \tilde s| \leq t_v / 4$ by the definition of $S_v$. Therefore:
\begin{align*}
    \widetilde{\ED}(Y_{v,\tilde s}, Y_{v,s'}) 
    &\leq n^{\order(1)} \cdot \ED(Y_{v, \tilde s}, Y_{v,s'}) \\
    &\leq n^{\order(1)} \cdot (\ED(Y_{v, s}, Y_{v,s'}) + \ED(Y_{v, s}, Y_{v, \tilde s}))  \\
    &\leq n^{\order(1)} \cdot (\ED(Y_{v,s}, Y_{v,s'}) + t_v/2),
\end{align*}
where second inequality is an application of the triangle inequality, and the last inequality follows since we can transform $Y_{v, \tilde s}$ into $Y_{v,s}$ by adding and removing $t_v/4$ symbols. A symmetric argument shows the claimed lower bound \raisebox{0pt}[0pt][0pt]{$\widetilde{\ED}(Y_{v,\tilde s}, Y_{v,s'}) \geq \ED(Y_{v,s}, Y_{v,s'}) - t_v / 2$}.

Next we analyze the running time. For the preprocessing, we compute~$|S_v| \cdot |S| = \Order^*(k^2/t_v)$ many approximations for each node $v$, each in time $|Y_v|^{1+\order(1)}$. We can bound~$\Ex(1/t_v) = n^{\order(1)} / k$ by \cref{lem:expected-precision}, so the expected total time is $\sum_v |Y_v|^{1+\order(1)} n^{\order(1)} \cdot k = k n^{1+\order(1)}$. Answering a query takes constant time since we only need to compute $s^*$ as stated above and perform a constant-time lookup.
\end{proof}

\subsection{Faster Shifted-Distance Queries} \label{sec:proofs:sec:fast-preprocessing}
In this section we show how to improve the preprocessing time for \ShiftedDistance{} queries to $n^{1 + o(1)}$. We thank the anonymous reviewer who suggested this improvement. The key technical tool to obtain this improvement is the following result by Andoni and Onak~\cite{AndoniO12}.

\begin{theorem}[Embedding Substrings into $\ell_1$~\cite{AndoniO12}]\label{thm:embedding-ao}
Let $X$ be a string of length $n$. Then, for each integer $m$ of the form $m = \lfloor n/B^i \rfloor$ for $0 \leq i \leq \log_B(n)$ and $B = 2^{\Theta(\sqrt{\log n \log \log n})}$, we can embed all length-$m$ substrings $X_0,\dots,X_{n-m}$ of $X$ into vectors $v_0^m,\dots,v_{n-m}^m$ of dimension $n^{o(1)}$ such that for every $j, j' \in \rangezero{n-m+1}$, with high probability it holds that
\begin{equation*}
    \ED(X_j, X_{j'}) \leq \|v_j^m - v_{j'}^m \|_1 \leq n^{o(1)} \cdot \ED(X_j, X_{j'}).
\end{equation*}
The time to compute all vectors is $n^{1 + o(1)}$.
\end{theorem}

The main application of this embedding in~\cite{AndoniO12} is an $n^{o(1)}$-approximation for the edit distance of two length-$n$ strings in time $n^{1+o(1)}$ (in~\cite{AndoniO12}, \cref{thm:embedding-ao} is applied to the concatenation of two strings to compute the $\ell_1$-distance between their corresponding vectors). We remark that the guarantee of the embedding in~\cref{thm:embedding-ao} is \emph{non-oblivious}, in the sense that the algorithm needs to have access to all the substrings it is embedding. In particular, this means that it cannot be directly applied in the two-sided preprocessing setting where we would like to embed the strings separately.

\lemshifteddistancequeries*
\begin{proof}
We assume without loss of generality that $n = |Y| = |X|$ is a power of $B$ (we can pad both $X$ and $Y$ so that this holds, which has no impact on the running time or the approximation guarantee of our algorithms). We apply~\cref{thm:embedding-ao} to the string $Y$ and store all the $n^{1+o(1)}$ embedded vectors. To answer a query $\ShiftedDistance(Y, v, s, s')$, note that $Y_{v,s}$ and $Y_{v,s'}$ are substrings of $Y$ of length $m := n/B^{\depth(v)}$. Thus, the $\ell_1$ distance between their corresponding vectors $v_{j(v,s)}^m, v_{j(v,s')}^m$ given by~\cref{thm:embedding-ao} gives the desired approximation (with no additive error). Since these vectors have dimension $n^{o(1)}$, the time to compute $\|v_{j(v,s)}^m - v_{j(v,s')}^m\|_1$ and answer the query is $n^{o(1)}$, as desired.
\end{proof}

\subsection{Pruning Rule for Preprocessed Strings} \label{sec:overview:sec:pruning-rule}
In this section we analyze \cref{alg:pruning-rule}. We always assume that either only $Y$ or both $X$ and $Y$ have been preprocessed by \cref{lem:matching-queries,lem:shifted-distance-queries} to efficiently answer \Matching{} and \ShiftedDistance{} queries.

\begin{algorithm}[t]
\caption{} \label{alg:pruning-rule}
\begin{algorithmic}[1]
\Input{Strings $X$ (un- or preprocessed) and $Y$ (preprocessed), a node $v$ in the precision tree}
\Output{An approximation $\widetilde\ED(X_v, Y_{v, s})$ of $\ED(X_v, Y_{v, s})$ for all $s \in S_v$, or \textsc{Fail}}
\medskip
\If{$\Matching(X, Y, v) = \Close[s^*]$} \label{alg:pruning-rule:line:matching-query}
    \State\Return $\widetilde\ED(X_v, Y_{v, s}) = \ShiftedDistance(Y, v, s^*, s)$ for all $s \in S_v$ \label{alg:pruning-rule:line:distance-query}
\Else
    \State\Return \textsc{Fail} \label{alg:pruning-rule:line:fail}
\EndIf
\end{algorithmic}
\end{algorithm}

\begin{lemma}[Correctness of \cref{alg:pruning-rule}] \label{lem:pruning-rule-correctness}
Whenever \cref{alg:pruning-rule} does not return \textsc{Fail}, it returns approximations~\raisebox{0pt}[0pt][0pt]{$\widetilde\ED(X_v, Y_{v, s})$} satisfying with high probability
\begin{equation*}
    \ED(X_v, Y_{v, s}) - t_v \leq \widetilde\ED(X_v, Y_{v, s}) \leq n^{\order(1)} \cdot (\ED(X_v, Y_{v, s}) + t_v).
\end{equation*}
\end{lemma}
\begin{proof}
\cref{alg:pruning-rule} only succeeds if the \Matching{} query in \cref{alg:pruning-rule:line:matching-query} successfully identified some shift $s^*$ with $\HD(X_v, Y_{v, s^*}) \leq t_v / 2$ by \cref{lem:matching-queries}. In this case, we report $\widetilde\ED(X_v, Y_{v, s}) := \ShiftedDistance(Y, v, s^*, s)$. By \cref{lem:shifted-distance-queries}, this is an approximation of $\ED(Y_{v, s^*}, Y_{v, s})$ with additive error $t_v / 2$ and multiplicative error $n^{\order(1)}$. Combining both facts, and using the triangle inequality we obtain that
\begin{align*}
    \widetilde\ED(Y_{v, s^*}, Y_{v, s})
    &\leq n^{\order(1)} \cdot (\ED(Y_{v, s^*}, Y_{v, s}) + t_v / 2) \\
    &\leq n^{\order(1)} \cdot (\ED(X_v, Y_{v, s}) + \ED(X_v, Y_{v, s^*}) + t_v / 2) \\
    &\leq n^{\order(1)} \cdot (\ED(X_v, Y_{v, s}) + t_v),
\end{align*}
and
\begin{align*}
    \widetilde\ED(Y_{v, s^*}, Y_{v, s})
    &\geq \ED(Y_{v, s^*}, Y_{v, s}) - t_v / 2 \\
    &\geq \ED(X_v, Y_{v, s}) - \ED(X_v, Y_{v, s^*}) - t_v / 2 \\
    &\geq \ED(X_v, Y_{v, s}) - t_v. \qedhere
\end{align*}
\end{proof}

\begin{lemma}[Running Time of \cref{alg:pruning-rule}] \label{lem:pruning-rule-time}
If only the string $Y$ is preprocessed, then \cref{alg:pruning-rule} runs in expected time $\Order^*(|X_v|/t_v + k/t_v)$. If both strings $X, Y$ are preprocessed, then \cref{alg:pruning-rule} runs in expected time $\Order^*(k / t_v)$.
\end{lemma}
\begin{proof}
The expected running time of the \Matching{} query is bounded by $\widetilde\Order(|X_v| / t_v)$ (for one-sided preprocessing) or by $\Order(1)$ (for two-sided preprocessing). The running time of a single \ShiftedDistance{} query is bounded by $n^{o(1)}$, and we make~$|S_v| = \widetilde\Order(k / t_v)$ \ShiftedDistance{} queries. Hence, the total expected time is~\makebox{$\Order^*(|X_v| / t_v + k / t_v)$} (for one-sided preprocessing) or $\Order^*(k / t_v)$ (for two-sided preprocessing).
\end{proof}

\begin{lemma}[Efficiency of \cref{alg:pruning-rule}] \label{lem:pruning-rule-efficiency}
For any two strings $X, Y$, there are at most~$\Order(k \log n)$ many nodes $v$ in the precision tree for which \cref{alg:pruning-rule} fails, assuming that $\ED(X, Y) \leq k$.
\end{lemma}
\begin{proof}
It suffices to argue that there are at most $ \Order(\ED(X, Y) \log n)$ nodes for which the \Matching{} query in \cref{alg:pruning-rule:line:matching-query} fails. We start with a fixed level in the precision tree, and enumerate all nodes on that level as~$v_1, \dots, v_m$. By definition we have that $X = \bigcirc_{i=1}^m X_{v_i}$, hence there exist indices~$0 = j_0 < \dots < j_m = n$ such that $X_{v_i} = X \range{j_i}{j_{i+1}}$. Now consider an optimal alignment $A : \set{0, \dots, n} \to \set{0, \dots, n}$ between $X$ and~$Y$. In particular, $A$ satisfies
\begin{equation*}
    \ED(X, Y) = \sum_{i=1}^m \ED(X \range{j_i}{j_{i+1}}, Y \range{A(j_i)}{A(j_{i+1})}).
\end{equation*}
There can be at most $\ED(X, Y) \leq k$ many nonzero terms in the sum, and we claim that every zero term corresponds to a node $v_i$ for which \cref{alg:pruning-rule} succeeds. Indeed, for any zero term we have $X \range{j_i}{j_{i+1}} = Y \range{A(j_i)}{A(j_{i+1})}$ and therefore $X_{v_i} = Y_{v_i, s^*}$ where $s^* = A(j_i) - j_i$. It remains to argue that $|s^*| \leq k$ (since otherwise the index $s^*$ would be out-of-bounds and could not be detected by a \Matching{} query). To see this, observe that
\begin{equation*}
    |s^*| \leq \ED(X \range{j_0}{j_i}, Y \range{A(j_0)}{A(j_i)}) \leq \ED(X, Y) \leq k,
\end{equation*}
exploiting again that $A$ is an optimal alignment.

Finally, recall that there are only $\log_B(n) \leq \log n$ levels in the precision tree, hence the total number of nodes for which \cref{alg:pruning-rule} fails is bounded by $\Order(k \log n)$.
\end{proof}

\subsection{The Complete Algorithm} \label{sec:overview:sec:main}
Our complete \cref{alg:main} is essentially what we described in \cref{sec:overview:sec:main} with the additional pruning rule of applying \cref{alg:pruning-rule}. We start with the correctness of \cref{alg:main}. We need the following lemma, which has previously been referred to as a \emph{Precision Sampling Lemma}~\cite{AndoniKO10,AndoniKO11,Andoni17,BringmannCFN22}. The lemma was introduced by Andoni, Krauthgamer and Onak~\cite{AndoniKO10}, and was refined in~\cite{AndoniKO11,Andoni17}.

Intuitively, the lemma serves the following purpose: For fixed numbers $a_1, \dots, a_B$, say that we have access to approximations $\widetilde a_1, \dots, \widetilde a_B$ with multiplicative error $\alpha$ and additive approximation error $\beta$. Then we can naively approximate $\sum_i a_i$ by $\sum_i \widetilde a_i$ with multiplicative error $\alpha$ and additive error $B \cdot \beta$. The Precision Sampling Lemma states that the blow-up in the additive error can be avoided if the approximations $\widetilde a_i$ instead have additive error $\beta \cdot u_i$ for some \emph{non-uniformly sampled precisions $u_i$.}   

\begin{algorithm}[t]
\caption{} \label{alg:main}
\begin{algorithmic}[1]
\Input{Strings $X$ (un- or preprocessed), $Y$ (preprocessed), a node $v$ in the precision tree}
\Output{An approximation $\widetilde\ED(X_v, Y_{v, s})$ of $\ED(X_v, Y_{v, s})$ for all $s \in S_v$}
\medskip
\If{\cref{alg:pruning-rule} succeeds and reports $\widetilde\ED(X_v, Y_{v, s})$ for all $s \in S_v$} \label{alg:main:line:prune-condition}
    \State\Return $\widetilde\ED(X_v, Y_{v, s})$ for all $s \in S_v$ \label{alg:main:line:prune}
\EndIf
\medskip
\If{$v$ is a leaf} \label{alg:main:line:leaf-condition}
    \State\Return $\ED(X_v, Y_{v, s})$ for all $s \in S_v$ \label{alg:main:line:leaf}
\EndIf
\For{all children $v_i$ of $v$} \label{alg:main:line:iter-children}
    \State Recursively compute $\widetilde\ED(X_{v_i}, Y_{v_i, s_i})$ for all $s_i \in S_{v_i}$ \label{alg:main:line:recursion}
    \State Compute $\widetilde a_{i, s} = \min_{s_i \in S_{v_i}} \widetilde\ED(X_{v_i}, Y_{v_i, s_i}) + 2 \cdot |s - s_i|$ for all $s \in S_v$ \label{alg:main:line:range-minimum}
\EndFor
\State\Return $\alg{Recover}(\widetilde a_{1, s}, \dots, \widetilde a_{B, s}, u_{v_1}, \dots, u_{v_B})$ for all $s \in S_v$ \label{alg:main:line:return}
\end{algorithmic}
\end{algorithm}

\begin{restatable}[Precision Sampling Lemma~{{{\cite{Andoni17}}}}]{lemma}{lempsl} \label{lem:psl}
Fix parameters $\delta > 0$, $\alpha \geq 1$ and $\beta \geq 0$. Let~\makebox{$a_1, \dots, a_B \geq 0$} be reals, and independently sample~$u_1, \dots, u_B \sim \Exp(\Order(\log(\delta^{-1})))$ (for some sufficiently large hidden constant). There is an $\Order(B \log(\delta^{-1}))$-time algorithm \alg{Recover} satisfying for all $\widetilde a_1, \dots, \widetilde a_B$, with success probability at least $1 - \delta$:
\begin{itemize}
\item If $\widetilde a_i \geq \frac1\alpha \cdot a_i - \beta \cdot u_i$ for all~$i$, then $\alg{Recover}(\widetilde a_1, \dots, \widetilde a_B, u_1, \dots, u_B) \geq \frac{1}{2\alpha} \cdot \sum_i a_i - \beta$.
\item If $\widetilde a_i \leq \alpha \cdot a_i + \beta \cdot u_i$ for all~$i$, then $\alg{Recover}(\widetilde a_1, \dots, \widetilde a_B, u_1, \dots, u_B) \leq 2 \alpha \sum_i a_i + \beta$.
\end{itemize}
\end{restatable}

\begin{lemma}[Correctness of \cref{alg:main}] \label{lem:main-correctness}
\cref{alg:main} computes values $\widetilde\ED(X_v, Y_{v, s})$ (for all $s \in S_v$) satisfying the following bounds with high probability:
\begin{itemize}
\item $\widetilde\ED(X_v, Y_{v, s}) \geq n^{-\order(1)} \cdot \ED(X_v, Y_{v, s}) - t_v$, and
\item $\widetilde\ED(X_v, Y_{v, s}) \leq n^{\order(1)} \cdot (\ED(X_v, Y_{v, s}) + t_v)$ assuming that $\ED(X_v, Y_{v, s}) \leq k$ and $2|s| \leq k$.
\end{itemize}
\end{lemma}

\begin{proof}
The recursion of \cref{alg:main} terminates in one of three cases: If $v$ is a leaf, then the output is exact and the claim is obvious. If $v$ is directly solved by \cref{alg:pruning-rule}, then by the guarantee of \cref{lem:pruning-rule-correctness} the claim is true. It remains to analyze the case when the algorithm recurs, so assume that $v$ is an internal node with children $v_1, \dots, v_B$. We prove the lower and upper bounds separately.

\proofsubparagraph{Lower Bound.}
Fix some shift $s \in S_v$, and let $s_1, \dots, s_B$ be the corresponding shifts selected in \cref{alg:main:line:range-minimum} of the algorithm. We prove that $\widetilde\ED(X_v, Y_{v, s}) \geq \ED(X_v, Y_{v, s}) / \alpha(\height(v)) - t_v$ by induction, where $\alpha(\height(v)) = 2^{\height(v)}$. To this end, we apply the Precision Sampling Lemma with $a_i = \ED(X_{v_i}, Y_{v_i, s_i}) + 2|s - s_i|$ and the following parameters:
\begin{itemize}
\item $\delta = 1 / \poly (n)$,
\item $\alpha = \alpha(\height(v) - 1)$,
\item $\beta = t_v$.
\end{itemize}
Recall that the algorithm computes $\widetilde a_{i, s} = \widetilde\ED(X_{v_i}, Y_{v_i, s_i}) + 2|s - s_i|$, where the approximation $\widetilde\ED(X_{v_i}, Y_{v_i, s_i})$ is computed recursively. Hence, it satisfies $\widetilde a_{i, s} \geq a_i / \alpha - \beta \cdot u_{v_i}$ by the induction hypothesis (recall that $t_{v_i} \leq t_v \cdot u_{v_i}$). Then $\alg{Recover}(\widetilde a_{1, s}, \dots, \widetilde a_{B, s}, u_{v_1}, \dots, u_{v_B})$ computes, with high probability, a number satisfying
\begin{align*}
    \alg{Recover}(\cdot)
    &\geq \frac{1}{2\alpha(\height(v) - 1)} \cdot \left(\sum_{i=1}^B a_i\right) - \beta \\
    &= \frac{1}{2\alpha(\height(v) - 1)} \cdot \left(\sum_{i=1}^B \ED(X_{v_i}, Y_{v_i, s_i}) + 2|s - s_i|\right) - t_v \\
    &\geq \frac{1}{2\alpha(\height(v) - 1)} \cdot \ED(X_v, Y_{v, s}) - t_v \\
    &\geq \frac{1}{\alpha(\height(v))} \cdot \ED(X_v, Y_{v, s}) - t_v,
\end{align*}
where in the third step, we applied the lower bound from \cref{lem:divide-and-conquer} and in the last step we used the definition of $\alpha(\cdot)$. Finally, recall that $\height(T) \leq \log_B(n)$ and $B = \exp(\widetilde \Theta(\sqrt{\log n}))$. Hence the total multiplicative error is bounded by $\alpha(\log_B(n)) \leq 2^{\log_B(n)} = n^{\order(1)}$.

\proofsubparagraph{Upper Bound.}
This proof is similar to the previous paragraph, but requires a more careful application of \cref{lem:divide-and-conquer}. We prove by induction the algorithm computes an approximation \raisebox{0pt}[0pt][0pt]{$\widetilde\ED(X_v, Y_{v, s}) \leq \alpha(\height(v)) \cdot (\ED(X_v, Y_{v, s}) + t_v)$} where this time the multiplicative error is bounded by $\alpha(\height(v)) \leq \Order(B) \cdot 2^{\Order(\height(v))}$, provided that $\ED(X_v, Y_{v, s}) \leq k \cdot 3^{\depth(v)}$ and~$2|s| \leq k \cdot 3^{\depth(v)}$. Note that this implies the lemma statement.

Throughout, fix some shift $s \in S_v$. The idea is to first use \cref{lem:divide-and-conquer} to find ``optimal'' shifts $s_1^*, \dots, s_B^* \in \Int$, which we use for the recursive computation. Unfortunately these shifts~$s_i^*$ may not fall into the restricted set of feasible shifts $S_{v_i}$. We therefore argue that picking shifts $s_i \in S_{v_i}$ closest possible to~$s_i^*$ is sufficient to obtain the claimed guarantee. Formally, by \cref{lem:divide-and-conquer} there exist shifts~$s_1^*, \dots, s_B^*$ satisfying the following two properties:
\begin{gather}
    \sum_{i=1}^B \ED(X_{v_i}, Y_{v_i, s_i^*}) \leq 2\ED(X_v, Y_{v, s}), \label{lem:main-upper-bound:eq:ed} \\
    2|s - s_i^*| \leq \ED(X_v, Y_{v, s}). \label{lem:main-upper-bound:eq:ed-shift}
\end{gather}
We now pick $s_1 \in S_{v_1}, \dots, s_B \in S_{v_B}$ to be the closest values to the optimal shifts $s_1^*, \dots, s_B^*$. As a first insight, observe that:
\begin{equation*}
    2|s_i^*| \leq 2|s - s_i^*| + 2|s| \leq \ED(X_v, Y_{v, s}) +2|s| \leq 2k \cdot 3^{\depth(v)}.
\end{equation*}
Recall that we set $S = \set{-k\cdot 3^{\log_B(n)}, \dots, k \cdot 3^{\log_B(n)}}$ and we therefore have~\makebox{$s_i^* \in S$}. It follows that $|s_i - s_i^*| \leq t_{v_i} / 2$ and thus
\begin{equation*}
    2|s_i| \leq 2|s_i^*| + t_{v_i} \leq 2|s_i^*| + k \leq 2k \cdot 3^{\depth(v)} + k \leq k \cdot 3^{\depth(v_i)}.
\end{equation*}
Next, we claim that $\ED(X_{v_i}, Y_{v_i, s_i}) \leq k \cdot 3^{\depth(v_i)}$, which we will use to guarantee that the recursive calls of the algorithm succeed. Indeed, we have that
\begin{equation*}
    \ED(X_{v_i}, Y_{v_i, s_i})
    \leq \ED(X_{v_i}, Y_{v_i, s_i^*}) + t_{v_i}
    \leq 2 \cdot \ED(X_v, Y_{v, s}) + k
    \leq k \cdot 3^{\depth(v_i)}.
\end{equation*}

We claim that if the algorithm was to choose the shifts $s_i$ specified in the previous paragraph in \cref{alg:main:line:range-minimum}, then the output is bounded as claimed. (This is sufficient, since \cref{alg:main:line:range-minimum} in fact minimizes over all possible shifts $s_i$.) In this case, we inductively have that
\begin{align*}
    \widetilde a_{i, s}
    &\leq \widetilde\ED(X_{v_i}, Y_{v_i, s_i}) + 2|s - s_i| \\
    &\leq \alpha(\height(v) - 1) \cdot (\ED(X_{v_i}, Y_{v_i, s_i}) + t_{v_i}) + 2|s - s_i| \\
    &\leq \alpha(\height(v) - 1) \cdot (\ED(X_{v_i}, Y_{v_i, s_i^*}) + 3 t_{v_i}) + 2|s - s_i^*|.
\end{align*}
In the last step we used that $\ED(X_{v_i}, Y_{v_i, s_i^*})$ differs from $\ED(X_{v_i}, Y_{v_i, s_i})$ by an additive error of $t_{v_i}$, and the same is true for $2|s - s_i^*|$ and $2|s - s_i|$. Next, we apply the Precision Sampling Lemma with $a_i = \ED(X_{v_i}, Y_{v_i, s_i^*}) + 2 \alpha^{-1} |s - s_i^*|$ and parameters
\begin{itemize}
\item $\delta = 1 / \poly (n)$,
\item $\alpha = \alpha(\height(v) - 1)$,
\item $\beta = \alpha t_v$.
\end{itemize}
Recall that $t_{v_i} = t_v \cdot u_{v_i} / 3$, thus by definition we have that $\widetilde a_{i, s} \leq \alpha \cdot a_i + \beta \cdot u_{v_i}$. Therefore, the Precision Sampling Lemma states that with high probability the recovery algorithm returns
\begin{align*}
    \alg{Recover}(\cdot)
    &\leq 2\alpha \cdot \left(\sum_{i=1}^B a_i\right) + \beta \\
    &\leq 2\alpha \left(\sum_{i=1}^B \ED(X_{v_i}, Y_{v_i, s_i^*})\right) + 2 \left(\sum_{i=1}^B 2|s - s_i^*|\right) + \alpha t_v \\
    &\leq 4\alpha \cdot \ED(X_v, Y_{v, s}) + 2B \cdot \ED(X_v, Y_{v, s}) + \alpha t_v \\
    &\leq 4(\alpha + B) \cdot (\ED(X_v, Y_{v, s}) + t_v) \\
    &\leq \alpha(\height(v)) \cdot (\ED(X_v, Y_{v, s}) + t_v),
\end{align*}
where for the third inequality we applied the bounds in \eqref{lem:main-upper-bound:eq:ed} and~\eqref{lem:main-upper-bound:eq:ed-shift} and in the last step we used the definition of $\alpha(\height(v)) = \Order(B) \cdot 2^{\Order(\height(v))}$ (for sufficiently large hidden constants). Since the tree has height $\log_B(n)$ and $B = \exp(\widetilde \Theta(\sqrt{\log n}))$, the overall approximation factor is bounded by $\alpha(\log_B(n)) = \Order(B) \cdot 2^{\Order(\log_B(n))} = n^{\order(1)}$, as claimed.
\end{proof}

Next, we analyze the running time of \cref{alg:main}. It turns out that the bottleneck is the computation in \cref{alg:main:line:range-minimum}, and a naive implementation would be too slow for our purposes. For that reason, we use the following lemma for an improved implementation of \cref{alg:main:line:range-minimum}. The lemma is a generalization of~\cite[Lemma~10]{BringmannCFN22}.

\begin{lemma}[Range Minimum Problem] \label{lem:range-minimum}
Let $T, T'$ be sets of integers (given in sorted order), and let~$(b_{s'})_{s' \in T'}$ be given. There is an $\Order(|T| + |T'|)$-time algorithm to compute $(a_s)_{s \in T}$ defined by
\begin{equation*}
    a_s = \min_{s' \in T'}\, b_{s'} + 2 \cdot |s - s'|.
\end{equation*}
\end{lemma}
\begin{proof}
The idea is to compute auxiliary values
\begin{gather*}
    a^{\leq}_s = \min_{\substack{s' \in T'\\s' \leq s}} b_{s'} + 2s - 2s', \\
    a^{\geq}_s = \min_{\substack{s' \in T'\\s' \geq s}} b_{s'} - 2s + 2s',
\end{gather*}
as then returning $a_s = \min(a^{\leq}_s, a^{\geq}_s)$ is correct. We show how to compute $a^{\leq}_s$ (for all $s$); the values $a^{\geq}_s$ are symmetric. Let $T = \set{s_1 < \dots < s_{|T|}}$. We evaluate the base case $a^{\leq}_{s_1}$ naively. We then compute $a^{\leq}_{s_i}$ for all $i = 2, \dots, |T|$ as follows:
\begin{equation*}
    a^{\leq}_{s_i} = \min\left\{a^{\leq}_{s_{i-1}} + 2s_i - 2s_{i-1}, \min_{s_{i-1} < s' \leq s_i} b_{s'} + 2s_i - 2s' \right\}\!.
\end{equation*}
For the correctness, we distinguish two cases. Let $s' \leq s_i$ be the index which attains the minimum in the definition of $a^{\leq}_{s_i}$. On the one hand, if $s' \leq s_{i-1}$, then $a^{\leq}_{s_i} = a^{\leq}_{s_{i - 1}} + 2s_i - 2s_{i-1}$ and thus the first term in the minimum is correct. On the other hand, if $s' > s_{i-1}$, then the second term in the minimum is correct by definition. In order to compute $(a^{\leq}_s)_s$ we sweep from left to right over all values in $T$ and $T'$ exactly once, hence the running time can be bounded by~$\Order(|T| + |T'|)$.
\end{proof}

\begin{lemma}[Running Time of \cref{alg:main}] \label{lem:main-time}
Assume that $\ED(X, Y) \leq k$. If only the string~$Y$ is preprocessed, then \cref{alg:main} runs in expected time $n^{1+\order(1)} / k + k n^{\order(1)}$. If both strings~$X, Y$ are preprocessed, then \cref{alg:main} runs in expected time $k n^{\order(1)}$.
\end{lemma}
\begin{proof}
We first bound the running time of a single execution of \cref{alg:main} (ignoring the cost of recursive calls). The computation in \cref{alg:main:line:leaf-condition,alg:main:line:leaf} merely compares single characters and therefore takes time $|S_v| = \Order^*(k / t_v)$. By~\cref{lem:pruning-rule-time}, the call to \cref{alg:pruning-rule} takes expected time \makebox{$\Order^*(|X_v| / t_v + k / t_v)$} (for one-sided preprocessing) or $\Order^*(k / t_v)$ (for two-sided preprocessing). Each iteration of the loop in \crefrange{alg:main:line:iter-children}{alg:main:line:range-minimum} is dominated by the computation in \cref{alg:main:line:range-minimum} (ignoring the recursive calls in \cref{alg:main:line:recursion}). Using \cref{lem:range-minimum} with~$T = S_v$ and~\makebox{$T' = S_{v_i}$}, this step takes time $\Order(|S_v| + |S_{v_i}|)$, and thus the loop runs in time~\makebox{$\Order(B \cdot |S_v| + \sum_i |S_{v_i}|) = \Order^*(k / t_v + \sum_i k / t_{v_i})$}. In all of these bounds we can bound~$1/t_v$ by~$n^{\order(1)} / k$ in expectation, according to \cref{lem:expected-precision}, so the total expected time per node becomes~\makebox{$n^{\order(1)} \cdot (|X_v| / k + 1)$} (for one-sided preprocessing) or~$n^{\order(1)}$ (for two-sided preprocessing).

To account for the recursive calls, we first use \cref{lem:pruning-rule-efficiency} to bound the number of recursive calls by $\Order(k \log n)$. Let $v_1, \dots, v_m$ (with $m = \Order(k \log n)$) denote all nodes for which \cref{alg:main} is recursively called. Then the expected running time for one-sided preprocessing can be bounded by
\begin{align*}
    n^{\order(1)} \cdot \sum_{i=1}^m \left(\frac{|X_{v_i}|}k + 1\right)
    &\leq k n^{\order(1)} + n^{\order(1)} \cdot \sum_{d=1}^{\log_B(n)} \sum_{\substack{i=1\\\depth(v_i) = d}}^m \frac{|X_{v_i}|}k \\
    &\leq k n^{\order(1)} + n^{\order(1)} \cdot \sum_{d=1}^{\log_B(n)} \frac nk \\
    &\leq k n^{\order(1)} + \frac{n^{1+\order(1)}}k.
\end{align*}
Here we used that across any level in the precision tree, the strings $X_v$ form a partition of $X$ into consecutive substrings. In the same way we can bound the running time for two-sided preprocessing by $k n^{\order(1)}$.
\end{proof}

\subsection{Proof of the Main Theorems} \label{sec:overview:sec:assemble}
We are finally ready to prove our main theorems.
\begin{proof}[Proof of \cref{thm:one-sided,thm:two-sided}]
We assume that either only $Y$ (for \cref{thm:one-sided}) or both $X$ and~$Y$ (for \cref{thm:two-sided}) are preprocessed by \cref{lem:matching-queries,lem:shifted-distance-queries}. We will solve the $(k, K)$-gap edit distance problem, for some parameter $K$ to be picked later, by running \cref{alg:main} (with input $v = \Root(T)$ and $s = 0$, so in particular $2|s| \leq k$). For the correctness proof, we apply \cref{lem:main-correctness} to the following two cases:
\begin{itemize}
\item If $\ED(X, Y) \leq k$, then $\widetilde\ED(X, Y) \leq n^{\order(1)} \cdot (\ED(X, Y) + t) \leq n^{\order(1)} \cdot k$.
\item If $\ED(X, Y) \geq K$, then \raisebox{0pt}[0pt][0pt]{$\widetilde\ED(X, Y) \geq n^{-\order(1)} \cdot \ED(X, Y) - t \geq n^{-\order(1)} \cdot K$}.
\end{itemize}
By setting $K = k \cdot n^{\order(1)}$ for a sufficiently large subpolynomial factor, we can distinguish the two cases based on the outcome \raisebox{0pt}[0pt][0pt]{$\widetilde\ED(X, Y)$}.

Next, we analyze the running time. If $\ED(X, Y) \leq k$, then the algorithm runs in expected time~$n^{1+\order(1)}/k + k n^{\order(1)}$ or $k n^{\order(1)}$, respectively; see \cref{lem:main-time}. By Markov's inequality, the algorithm respects these time bounds with constant probability. We may therefore interrupt the algorithm after it exceeds its time budget and report~``$\ED(X, Y) \geq K$'' in case of an interruption. We can boost the success probability to~$1 - 1/\poly (n)$ by running the algorithm~$\Order(\log n)$ times in parallel and reporting the majority answer. (This also means that the preprocessing is repeated $\Order(\log n)$ times with independently sampled precision trees.)
\end{proof}

\bibliographystyle{plain}
\bibliography{refs}

\appendix
\section{\texorpdfstring{Proofs of \cref{lem:divide-and-conquer,lem:expected-precision}}{Missing Proofs}} \label{sec:divide-and-conquer}
In this section we provide proofs of \cref{lem:divide-and-conquer,lem:expected-precision}.

\lemdivideandconquer*
\begin{proof}[Proof. Lower Bound]\enspace%
Let us write $Y_i = Y_{i, 0} = Y \range{j_{i-1}}{j_i}$ (in analogy to the notation~$X_i$). It is clear that $\ED(Y_{i, s_i}, Y_i) \leq 2|s_i|$ by deleting and inserting at most $|s_i|$ symbols. Therefore, by several applications of the triangle inequality we have
\begin{align*}
    \sum_{i=1}^B \ED(X_i, Y_{i, s_i}) + 2|s_i|
    &\geq \sum_{i=1}^B \ED(X_i, Y_{i, s_i}) + \ED(Y_{i, s_i}, Y_i) \\
    &\geq \sum_{i=1}^B \ED(X_i, Y_i) \\
    &\geq \ED(X, Y).
\end{align*}

\proofsubparagraph{Upper Bound.}
For the upper bound, let $A : \set{0, \dots, n} \to \set{0, \dots, n}$ denote an optimal alignment between $X$ and $Y$ (as defined in \cref{sec:preliminaries}). Then
\begin{equation}
    \ED(X, Y) = \sum_{i=1}^B \ED(X \range{j_{i-1}}{j_i}, Y \range{A(j_{i-1})}{A(j_i)}). \label{lem:divide-and-conquer:eq:alignment}
\end{equation}
We pick $s_i = A(j_{i-1}) - j_{i-1}$. The first step is to prove that \makebox{$2|s_i| \leq \ED(X, Y)$}, thereby proving the second item of the upper bound. To see this, we first express $\ED(X, Y)$ as the sum of two edit distances $\ED(X \range{0}{j_{i-1}}, Y\range{0}{A(j_{i-1})})$ and $\ED(X \range{j_{i-1}}{n}, Y\range{A(j_{i-1})}{n})$, using that $A$ is an optimal alignment. Now, since the edit distance of two strings $A, B$ is always at least $|\, |A| - |B| \,|$, we conclude that $\ED(X, Y) \geq 2|s_i|$.

The next and final part is to prove that $\sum_i \ED(X_i, Y_{i, s_i}) \leq 2\ED(X, Y)$.
\begin{align*}
    \ED(X_i, Y_{i, s_i})
    &= \ED(X \range{j_{i-1}}{j_i}, Y \range{A(j_{i-1})}{A(j_{i-1}) + j_i - j_{i-1}}) \\
    &\leq \ED(X \range{j_{i-1}}{j_i}, Y \range{A(j_{i-1})}{A(j_i)}) + |(A(j_i) - A(j_{i-1})) - (j_i - j_{i-1})| \\
    &\leq 2 \cdot \ED(X \range{i_{j-1}}{i_j}, Y \range{A(j_{i-1})}{A(j_i)}),
\end{align*}
where in the last step we again used that the edit distance between two strings is at least their length difference. It follows that
\begin{equation*}
    \sum_{i=1}^B \ED(X_i, Y_{i, s_i})
    \leq 2 \sum_{i=1}^B \ED(X \range{j_{i-1}}{j_i}, Y \range{A(j_{i-1})}{A(j_i)})
    \leq 2\ED(X, Y). \qedhere
\end{equation*}
\end{proof}

\lemexpectedprecision*
\begin{proof}
Recall that we assign $t_v = t_{\parent(v)} \cdot u_v / 3$ for all non-root nodes in the precision tree, where the samples $u_v$ are independent of each other and sampled from $\Exp(\lambda)$, the exponential distribution with parameter $\lambda = \Order(\log n)$. Recall that the exponential distribution has probability density function $f(x) = \lambda e^{-\lambda x}$, and thus
\begin{equation*}
    \Pr_{u \sim \Exp(\lambda)}(u \leq x) = \int_{u=0}^x \lambda e^{-\lambda x} du = 1 - e^{-\lambda x} \leq \lambda x.
\end{equation*}
We let $E$ denote the event that $u_v \geq 1/\poly(n)$ for all nodes $v$. For any specific node $v$ we have $u_v \geq 1/\poly(n)$ with high probability, and thus by a union bound $E$ happens with high probability as well.

Next, we prove that conditioned on $E$, the expectation of $1/u$ for $u \sim \Exp(\lambda)$ is small:
\begin{align*}
    \Ex_{u \sim \Exp(\lambda)} \left(1/u \mid E\right)
    &\leq \frac{1}{\Pr(E)} \cdot \int_{u=1/\poly(n)}^\infty \frac1u \cdot \lambda e^{-\lambda u} du \\
    &\leq \frac{1}{\Pr(E)} \cdot \left(\int_{u=1/\poly(n)}^1 \frac1u \cdot \lambda e^{-\lambda u} du + \int_{u=1}^\infty \frac1u \cdot \lambda e^{-\lambda u} du\right) \\
    &\leq \frac{1}{\Pr(E)} \cdot \Order(\lambda \log n + 1) \\
    &\leq \Order(\log^2 n).
\end{align*}

We now prove the claimed bound. Fix some node $v$ and let $\Root(T) = v_1, \dots, v_{\depth(v)} = v$ denote the root-to-$v$ path in the precision tree. Recall that $t_v = t \cdot (u_{v_1} / 3) \cdots (u_{v_{\depth(v)}} / 3)$. We have
\begin{align*}
    \Ex\left(\frac1{t_v}\;\middle|\;E\right)
    &= \frac1t \cdot \left(3 \cdot \Ex_{u \sim \Exp(\lambda)}\left(\frac{1}{u}\;\middle|\;E\right)\right)^{\depth(v)} \\
    &\leq \frac{(\log n)^{\Order(\depth(v))}}{t} \\
    &\leq \frac{n^{\order(1)}}t.
\end{align*}
For the last inequality, we used that $\depth(v) \leq \log_B(n) = \widetilde\Order(\sqrt{\log n})$, and therefore the overhead becomes $\exp(\widetilde\Order(\sqrt{\log n})) = n^{\order(1)}$.
\end{proof}

\end{document}